\documentclass[reqno,12pt]{article}
\usepackage{amsmath,amsfonts,amssymb,amsthm,amstext,amscd,eucal,xcolor}
\usepackage{color}
\usepackage[all]{xy}
\usepackage{hyperref}
\usepackage{epsfig}
\usepackage[active]{srcltx}
\makeatletter \@addtoreset{equation}{section}

\newcommand{\be}{\begin{equation}}
\newcommand{\ee}{\end{equation}}
\newcommand{\bea}{\begin{eqnarray}}
\newcommand{\eea}{\end{eqnarray}}
\newcommand{\bse}{\begin{subequations}}
\newcommand{\ese}{\end{subequations}}
\newcommand{\beqa}{\begin{eqnarray}}
\newcommand{\eeqa}{\end{eqnarray}}
\newcommand{\beqar}{\begin{eqnarray*}}
\newcommand{\eeqar}{\end{eqnarray*}}
\newcommand{\bi}{\begin{itemize}}
\newcommand{\ei}{\end{itemize}}
\newcommand{\bn}{\begin{enumerate}}
\newcommand{\en}{\end{enumerate}}
\newcommand{\fixme}[1]{\textbf{FIXME: }$\langle$\textit{#1}$\rangle$}

\newcommand{\ba}{\begin{array}}
\newcommand{\ea}{\end{array}}
\newcommand{\bc}{\begin{center}}
\newcommand{\ec}{\end{center}}
\newcommand{\nnr}{\nonumber \\}

\definecolor{darkgreen}{rgb}{0,0.3,0}
\definecolor{mgreen}{rgb}{0,0.6,0}
\definecolor{darkblue}{rgb}{0,0,0.3}
\definecolor{darkred}{rgb}{0.7,0,0}

\newtheorem{theorem}{Theorem}
\newtheorem{lemma}{Lemma}
\begin{document}

\begin{titlepage}

\begin{flushright}\vspace{-3cm}
{\small
IPM/P-2015/nnn \\
\today }\end{flushright}
\vspace{0.5cm}

\begin{center}
\centerline{{\fontsize{18pt}{12pt}\selectfont{\bf{Near Horizon Structure of Extremal Vanishing Horizon Black Holes}}}} \vspace{10mm}

\centerline{\fontsize{15pt}{12pt}\selectfont{\bf{S. Sadeghian\footnote{e-mail: s\_sadeghian@alzahra.ac.ir}$^{;a,b}$, M.M. Sheikh-Jabbari\footnote{e-mail:
jabbari@theory.ipm.ac.ir}$^{;a}$,
   M.H. Vahidinia\footnote{email: vahidinia@ipm.ir}$^{;a}$, H. Yavartanoo\footnote{e-mail:
yavar@itp.ac.cn }$^{;c}$}}}

\vspace{5mm}
\normalsize
{$^a$ \it School of Physics, Institute for Research in Fundamental
Sciences (IPM),\\ P.O.Box 19395-5531, Tehran, Iran}\\
\smallskip
{$^b$ \it Department of Physics, Azzahra University,   Tehran 19938-93973, Iran}\\
\smallskip
{$^c$ \it  State Key Laboratory of Theoretical Physics, Institute of Theoretical Physics,
Chinese Academy of Sciences, Beijing 100190, China.
 }
\end{center}
\begin{abstract}
\noindent
We study the near horizon structure of Extremal Vanishing Horizon (EVH) black holes,  extremal black holes with vanishing horizon area with a vanishing one-cycle on the horizon.  We construct the most general near horizon  EVH and near-EVH ansatz for the metric  and other fields, like dilaton and gauge fields which may be present in the theory.  We prove that (1) the near horizon EVH geometry for generic gravity theory in generic dimension has a three dimensional maximally symmetric subspace; (2) if the matter fields of the theory satisfy strong energy condition either this 3d part is  AdS$_3$, or the solution is a direct product of a locally 3d flat space and a $d-3$ dimensional part; (3) these results extend to the near horizon geometry of near-EVH black holes, for which the AdS$_3$ part is replaced with BTZ geometry. We present some specific near horizon EVH geometries in 3, 4 and 5 dimensions for which there is a classification. We also briefly discuss implications of these generic results for generic (gauged) supergravity theories and also for the thermodynamics of near-EVH black holes and the EVH/CFT proposal.

\end{abstract}


\end{titlepage}
\setcounter{footnote}{0}
\renewcommand{\baselinestretch}{1.05}  

\addtocontents{toc}{\protect\setcounter{tocdepth}{2}}
\tableofcontents


\section{Introduction and motivation}\label{sec-intro}

Classification of solutions to Einstein gravity theory in diverse dimensions, coupled to various matter fields has been an active area of research for many years (basically since the conception of General Relativity); e.g. see \cite{Hawking-Ellis, Stephani,Stationary-Black-Holes, ER-review}. Black holes constitute an important sector of these solutions and the corresponding classification and uniqueness theorems provide crucial pieces of information for both classical and quantum aspects of black hole physics. Black hole solutions may be defined as space-times with event horizons or spacetimes with (null) closed trapped surfaces \cite{Hawking-Ellis}. However, the more usual specification which is what we will use here, is stationary solutions with a compact, smooth Killing horizon. The Killing horizon is the null surface generated by the orbits of a Killing vector field $\xi_H$ which becomes null at the horizon.

Extremal black holes are a special class of black holes which have vanishing surface gravity (where $\xi_H$  also becomes the velocity vector of a geodesic congruence).  In this case one can show that the horizon is degenerate (there is no bifurcate horizon) \cite{Kay-Wald-review}. Extremal black holes have vanishing Hawking temperature and hence do not Hawking radiate; have the lowest possible mass with the given set of charges (as set by the extremality bound) and  may hence be viewed as ``ground states'' of the non-extremal black holes in the same class. Moreover, in supersymmetric gravity theories, one may show that all supersymmetric (BPS) black holes are necessarily extremal. For these reasons extremal black holes have been studied extensively, especially, given their ``stability'' features \cite{Ferrara-review, Sen-review}. They also provide a fruitful test ground for studying certain quantum aspects of black holes, in particular  black hole microstate counting \cite{Sen-review, NHEG-phase-space}.

Extremal black holes, also represent a very interesting feature: the near horizon geometry of generic extremal black holes give rise to a new class of (non-black hole) solutions and contain a 2d maximally symmetric space (generically an AdS$_2$ part). The Near Horizon Extremal Geometries (NHEG's) have generically  SL(2,R)$\times$U(1)$^n$ isometry. There are powerful existence and uniqueness theorems about such solutions \cite{KL-papers,KL-review}. In particular, there are uniqueness theorems for 4d and 5d Einstein-Maxwell-Dilaton (EMD) theories and also for solutions with SL(2,R)$\times$U(1)$^{d-3}$ isometry in a generic $d$ dimensional EMD theory \cite{KL-review,Hollands}. Moreover, one can show that  NHEG's enjoy a variant of laws of black hole mechanics \cite{Wald}, the laws of NHEG dynamics \cite{NHEG-mechanics}.

The  existence and uniqueness theorems for the NHEG mentioned above,  crucially use strict smoothness of the horizon. There are, however, interesting subclass of extremal black holes which do not have smooth horizons. Among them, there is a special class,  Extremal Vanishing Horizon (EVH) black holes, which have vanishing horizon area. Let us  denote the horizon area by $A_h$ and the Hawking temperature by $T$ and consider a ``near extremal'' ($T\to 0$) geometry. If going to extremal point is such that
\be\label{EVH-def-1}
A_h, T\to 0 \qquad \frac{T}{A_h}=finite\,
\ee
and that vanishing of the horizon area comes from the vanishing of a single (compact) cycle on the horizon, we are dealing with an EVH black hole \cite{EVH/CFT,First-law}.
Various aspects of the specific, given EVH black holes, their near-horizon geometry and a possible dual 2d CFT picture for the corresponding excitations, dubbed as EVH/CFT proposal \cite{EVH/CFT}, has been analyzed in several earlier papers \cite{Bar-Horo, EVH-2,AdS4-EVH,Japanese-EVH,BTZ-EVH,EVH-3,EVH-4,EVH-Ring,Hossein-HD,Hossein-4d,Hossein-5d}.

Through analysis of several examples of EVH black hole solutions in various theories in diverse dimensions, we have learned that
\begin{enumerate}
\item In the near horizon geometry of EVH black holes we find a ``pinching AdS$_3$'', that is an AdS$_3/Z_K$ orbifold with $K\to\infty$ \cite{BTZ-EVH}.  The near horizon EVH geometry exhibits (local) SO(2,2) isometry.
\item In the near horizon geometry of ``near-EVH'' case, the pinching AdS$_3$ is replaced by a ``pinching BTZ'' geometry.
Analysis of the first law of black hole thermodynamics for near-EVH black holes is compatible with the appearance of pinching BTZ in their near horizon limit \cite{EVH/CFT,First-law,EVH-4}.
\item EVH black holes can appear in any dimension $d\geq 3$. The simplest EVH black hole is massless BTZ \cite{BTZ-EVH}. They may also appear in the asymptotic flat, asymptotic de Sitter or AdS spaces.
\item The Brown-Henneaux central charge associated with this AdS$_3$ factor is proportional to $A_h/T$ ratio in the EVH limit.
\item The near horizon limit of (near) EVH geometry is presumably a decoupling limit \cite{EVH/CFT}.
\item EVH black holes can be static or stationary.  EVH black holes can be supersymmetric (in SUSY theories) \cite{EVH-2,AdS4-EVH}. Near-EVH black holes can be extremal.
\item Black rings can also become EVH, where they meet the corresponding black holes with the same charge  \cite{EVH-Ring}.
\end{enumerate}

Some of the above features, in particular items 1-3, seems to be quite generic and not limited to the specific examples studied. In this work we would like to explore how generic features 1-3 are.
In the 3d case the situation is easy to analyze, as we know the full set of (asymptotic AdS$_3$) black hole solutions, the BTZ family \cite{BTZ-EVH}. In the 4d case, within Einstein-Maxwell-Dilaton theory, two of us \cite{EVH/CFT}, proved items 1-3. In fact, in that paper we also proved items 4,5, that in the 4d case the near horizon limit is indeed a decoupling limit. In this work, we  mainly focus on  proving items 1-3 in generic Einstein-Maxwell-Scalar theories in diverse dimensions.

The organization of this paper is as follows. In section \ref{EVH-ansatz-section},  we adopt the Gaussian null coordinates for a generic black hole geometry and work out the most general metric which exhibits \eqref{EVH-def-1},  the ``EVH ansatz''. In section \ref{NHEVH-ansatz-section}, we take the near horizon limit over the EVH ansatz and provide the near-horizon EVH (NHEVH) ansatz. In section \ref{NHEVH-EOM-section}, we discuss general implications of Einstein equations and the smoothness of the energy momentum tensor on the near-horizon EVH ansatz. We also prove that for the matter fields satisfying strong energy conditions the near-horizon EVH geometries generically have an AdS$_3$ throat. In section \ref{NHEVH-Matter-field-section}, we strict ourselves to theories with scalar and Maxwell gauge fields and work out the most general ansatz for the scalar and gauge fields for an EVH solution. We show that energy momentum tensor for these theories fit well with the assumptions made in section \ref{NHEVH-EOM-section} and hence the theorems proved there apply for this class of theories. In section \ref{near-EVH-section}, we consider the most general ansatz for the near-EVH geometries and repeat analysis of section \ref{NHEVH-EOM-section} for the near-horizon limit of near-EVH solutions. We prove that the AdS$_3$ factor of the  near-horizon EVH geometries is replaced with a BTZ for near-EVH black holes.  In section \ref{examples-section}, we discuss explicit NHEVH solutions in 3, 4 and 5 dimensions for which we have a full classification. Section \ref{discussion-section} is devoted to the summary of our results and outlook.


\section{General ansatz for stationary EVH black holes}\label{EVH-ansatz-section}

Black hole solutions of a given theory of gravity are  described by a set of parameters usually defined at the (outer) horizon, the horizon (thermodynamical) properties, like the surface gravity or the Hawking temperature $T$ and horizon angular velocity or electromagnetic potentials, and the corresponding (thermodynamically conjugate) conserved charges, usually defined in the asymptotic region. There is also the entropy $S$ which is a conserved charge defined at the (bifurcate) horizon. The horizon properties are fixed and well-defined once we choose a specific reference frame in the asymptotic region \cite{Wald}. Each point in the black hole parameter space defines a black hole solution and one may study limits over this parameter space. The EVH black holes, as defined through \eqref{EVH-def-1} limit over the parameter space, is the specific class we will focus on.

We would like to write down the most general form of EVH black hole metric. To this end, as has been shown and used in the case of extremal black holes \cite{KL-papers,KL-review}, we also find it convenient to adopt the Gaussian null coordinates. For any stationary black hole, we can write the metric as \cite{KL-review}
\bea \label{401}
ds^2=2drdv+2r{\mathsf f}_i(r,y)dvdy^i-r {\mathsf F}(r,{y}) dv^2 + {\mathsf h}_{ij}(r,{y})dy^idy^j,\qquad i,j=1,2,\cdots,d-2\,,
\eea
where  $N=\partial_{v}$ is a Killing vector field which  becomes null at the outer horizon  located at $r=0$; $\partial_v$ is the vector field which creates the outer Killing horizon. The Gaussian null coordinates covers the region outside the Killing horizon and is constructed such that $\partial_r$ is a vector field which is null everywhere.  The Hawking temperature can be calculated using the surface gravity $\kappa$ which is defined by
\bea\label{eq00ch3}
\kappa N_{\mu}=N^{\nu} \nabla_{\nu} N_{\mu} =-N^{\nu} \nabla_{\mu} N_{\nu}=-\frac{1}{2}\nabla_{\mu}(N \cdot N)\;.
\eea
Taking $\mu=r $ component of  above equation and recalling that $N_{r}=g_{r v} N^v=g_{r v}=1$, $N\cdot N=g_{vv}$, we get
\bea\label{T}
\kappa=-\frac12 \partial_{r}g_{v v}\Big|_{r=0}\quad \Longrightarrow  \quad T=\frac{\kappa}{2\pi}=\frac{ {\mathsf F}(r=0,y)}{4\pi}.
\eea
Therefore, in the extremal limit ${\mathsf F}(r=0,y)$ vanishes. Moreover, since  surface gravity is a nonnegative constant over the Killing horizon,  ${\mathsf F}(r=0,y)\geq 0$ and is  independent of $y$.

To implement the definition of EVH \eqref{EVH-def-1}, we need to consider a ``near-extremal'' black hole and take the extremal limit. Let us parameterize out-of-extremality by a small positive parameter $\epsilon$. Assuming ${\mathsf F}(r,y)$ is an analytical function near the horizon, we can expand it to get
\bea\label{Fexpans}
 {\mathsf F}(r,y)=\epsilon { { F}}^{(1)}+r  { F}(y)+\cdots\,,
\eea
where $F^{(1)}$ is a nonnegative constant ($F^{(1)}\geq 0$).  Next, we need the horizon area $A_h$:
\bea\label{S}
A(r)\equiv \int_{const.\ r} \sqrt{\det {\mathsf h}} \; d^{d-2}y\,,\qquad A_h=A(r=0)\,.
\eea
Definition of the EVH black hole requires $A_h\sim \epsilon \to 0$, while  vanishing of $A_h$ is due to a vanishing one-cycle at $r=0$. That is, we require that ${\mathsf h}_{ij}$ has only one zero eigenvalue at $r=0$. Let us denote the corresponding eigenvector by $\partial_\phi$, and decompose $y^i$ directions as $(x^a,\phi)$, where $a=1,2,\cdots,d-3$. Assuming that  $\phi$ is an isometry direction, the horizon metric ${\mathsf h}$ can be decomposed as
\bea\label{h-ij}
{\mathsf h}_{ij}dy^idy^j={\mathsf G}(r,x) d\phi^2+2{\mathsf g}_a(r,x) d\phi dx^a+ \hat\gamma_{ab}(r,x) dx^adx^b,
\eea
where
\be\label{phi-x-metric-expansion}
{\mathsf g}_a|_{r=0}=\epsilon g_a^{(1)}(x^a) ,\quad {\mathsf G}|_{r=0}=\epsilon^2 G^{(2)}(x^a)\,.
\ee

The metric \eqref{401} together with \eqref{Fexpans}, \eqref{h-ij} and \eqref{phi-x-metric-expansion}, in the $\epsilon\to 0$ limit,  provides the metric ansatz for EVH black hole, which by construction satisfies \eqref{EVH-def-1}. If this metric ansatz (possibly together with other matter field configurations) is a solution to (Einstein) equations of motion, we have an EVH black hole solution. Note that $\epsilon$ in this EVH black hole solution is one of the parameters of the black hole solution family (e.g. determining angular momentum or charge) and we get an EVH black hole at the specific $\epsilon=0$ point in the black hole parameter space. Here, we do not check if this EVH black hole ansatz provides a solution; we assume that the EVH black hole solution exists. What we will do next is to study the near horizon limit of this EVH black hole.

\section{Near horizon limit of generic EVH black hole ansatz}\label{NHEVH-ansatz-section}

For the black hole solutions defined by metric \eqref{401} and its EVH limit, horizon is located at $r=0$. We therefore, define the  near horizon EVH geometry through
\bea\label{rigidscaling}
r\rightarrow \lambda r, \qquad v\rightarrow \frac{v}{\lambda},  \qquad \phi\rightarrow \frac{\phi}{\lambda}\;,\qquad \lambda\to 0\;.
\eea
With the above limit, a generic field of this solution $\Pi(r;x)$ will show a double expansion in powers of the ``near-EVH'' parameter $\epsilon$ and the near-horizon parameter $\lambda$:
\bea\label{double-expansion}
\Pi(r;x) &=& \lambda^s(\Pi^{(0)}(x)+r\lambda \Pi^{(1)}(x)+\cdots)\cr &=&\lambda^s(\Pi^{(0,0)}(x)+\epsilon\Pi^{(0,1)}(x)+r\lambda \Pi^{(1,0)}(x)+r\lambda\epsilon \Pi^{(1,1)}(x)+\cdots),
\eea
where the parameter $s$ denotes the leading power in $\lambda$ and  can be different for different fields. As we will discuss in section \ref{NHEVH-Matter-field-section}, the EVH conditions imply that for the scalar/dilaton fields $s=0$; for Maxwell gauge fields, depending on the component, $s$ can be $-1, 0$ or $1$. For the metric components, as we will show below, in the Gaussian null coordinates, $s$ can be $0,1,2$.

Taking the above $\lambda\to 0$ limit along with the EVH $\epsilon\to 0$ limit defined in the previous section, we get
\bea
\label{genericNH}
 ds^{2}\!\!\! &=&\!\!\!-r\left( \frac{\epsilon }{\lambda }{F^{(1)}} +rF
\right) dv^{2} +2r\left( \frac{\epsilon }{\lambda }{H}^{(1)}+rH
\right) d\phi dv  +\left( \frac{\epsilon ^{2}}{\lambda ^{2}}{
{G}^{(2)}} +\frac{\epsilon }{\lambda }r{G}^{(1)}+r^{2}G\right) d\phi ^{2}\nonumber \\
\!\!\!& +&\!\!\! 2drdv+ 2r f_a   dx^{a}dv   +2\left( \frac{\epsilon }{\lambda }{g}_{a}^{(1)}+rg_{a}\right) dx^{a}d\phi +\gamma _{ab}dx^{a}dx^{b}+{\mathcal{O}}(\lambda,\epsilon )\;.
\eea
where $\gamma_{ab}(x)=\hat\gamma_{ab}(r=0,x)$, and all the functions in the above metric are only functions of $x^a$. In the above we have used the fact that all  functions in the metric \eqref{401} around the horizon at $r=0$ are smooth and analytic.

Equation \eqref{genericNH} gives the most general near-horizon near-EVH metric, where proximity to the EVH is measured by $\epsilon$ and proximity to horizon  by $\lambda$. Therefore, while we are taking $
\epsilon,\lambda\to 0$ limit one can imagine three cases:
\begin{itemize}
\item
$\frac{\epsilon}{\lambda}\ll 1$. This case corresponds to  the near-horizon EVH (NHEVH) geometry. This is the case we analyze in sections \ref{NHEVH-EOM-section} and \ref{NHEVH-Matter-field-section}.
\item  $\frac{\epsilon}{\lambda}\sim 1$, corresponds to near-EVH near-horizon limit. We  view this case as ``excitations'' over the NHEVH geometry.  This case has some parameters/functions more
than the ``background'' NHEVH geometry. This is the case we analyze in section \ref{near-EVH-section}.
\item  $\frac{\epsilon}{\lambda}\gtrsim 1$ case would correspond to far from EVH case and we would not expect to get a well-defined near-horizon geometry. So, we will exclude the case.
\end{itemize}

In this section and section 4 we consider   ${\epsilon}\ll {\lambda}$ and once we established the existence of the NHEVH geometry as a solution to a gravity theory, we study the $\epsilon\sim \lambda$ case.
For  ${\epsilon}\ll {\lambda}$ we obtain the form of our \emph{near-horizon EVH (NHEVH) ansatz}
\bea\label{NHEVH}
ds^{2} =r^{2}\left[ -Fdv^{2}+G d\phi ^{2}+2H d\phi dv\right]  +2drdv
+2r\left[ f_{a} dx^{a}dv+g_{a} dx^{a}d\phi \right] +\gamma
_{ab} dx^{a}dx^{b}\;,
\eea
where all unknown functions are only function of $x^a$.
Some remarks about the above NHEVH ansatz are in order:
\begin{enumerate}
\item $\partial_v$ is a Killing vector field which is null at $r=0$ and is timelike elsewhere if $F>0$.
\item Determinant of the above metric is proportional to $-r^2G$. As we will see in the next section, this proportionality constant is positive and hence we learn that $G(x)>0$. Also, $\partial_\phi$ is a  Killing vector and $|\partial_\phi|^2=r^2G$. Hence $\partial_\phi$ remains spacelike everywhere.
\item As a result of the near-horizon limit and in the Gaussian null coordinates we have employed, $r$ dependence of all components of metric are fixed. In particular, metric has $r^2$ terms ($v,\phi$ part), $r$ terms (the off-diagonal $dv dx^a, d\phi dx^a$ terms) and the $r$ independent pieces, $\gamma_{ab}$ and $drdv$ terms.
\item As \eqref{rigidscaling} indicates, if the original $\phi$ coordinate (which parameterizes the vanishing one-cycle of the EVH black hole) had $2\pi$ periodicity,  the $\phi$ coordinate which appears in \eqref{NHEVH} has $2\pi\lambda$ periodicity. Since we are assuming $\phi$ to be a Killing direction its periodicity does not appear in the discussions of the equations of motion (which are local). Nonetheless, once we want to compare the NHEVH solution with the original EVH black hole, to stress this reduced range of $\phi$, we call $\phi$ as the \emph{pinching direction}.
\item Determinant of the NHEVH ansatz is proportional to $r^2$. This is compatible with our expectations, and construction of the EVH geometry that the horizon area of the constant $v$ part of metric is zero at $r=0$.
\item The $d\!-\!3$ dimensional part spanned by $x^a$ coordinates and the metric $\gamma_{ab}$, will be denoted by $\gamma_{d-3}$. For EVH black holes, upon which we focus in this work, $\gamma_{d-3}$ has finite volume.
\item One may reduce the $d$ dimensional Einstein Hilbert action over $\gamma_{d-3}$ and obtain a 3d gravity theory with some scalars and vector fields. We will comment more on this in the discussion section.
\item In the NHEVH ansatz we have used (and fixed) all diffeomorphisms along the 3d $r,v,\phi$ directions. However, the 3d-foliation preserving diffeomorphisms; i.e. $x^a\to h_a(x^b)$, can  still be used to simplify form of $\gamma_{ab}(x)$ metric.
\end{enumerate}

\section{Implications of Einstein equations for the NHEVH ansatz}\label{NHEVH-EOM-section}

In the previous section we gave the most general form of the near-horizon (near) EVH metric. In this section we impose the Einstein equations on this ansatz. We assume that the EVH black hole is a solution to Einstein equations
\be\label{Ein-Eq}
G_{\mu\nu}+\Lambda g_{\mu\nu}=8\pi G_N T_{\mu\nu}
\ee
where $T_{\mu\nu}$ is the energy momentum tensor of the matter fields in our theory. We will use the units where $8\pi G_N=1$. In this section, although we choose the gravity part of our action to be Einstein-Hilbert type, we do not make any specific choice for the matter part of the theory. Here, we first analyze generic properties and features of $T_{\mu\nu}$ in the near horizon limit and study the implications of the smoothness of $T_{\mu\nu}$  on the NHEVH solutions. Assuming specific behavior for some components of $T_{\mu\nu}$ we show that  the NHEVH solution has generically a 3d part which is locally maximally symmetric.
We then prove that if $T_{\mu\nu}$ satisfies strong energy condition the NHEVH solution should necessarily have an AdS$_3$ throat, or takes the form of direct product of 3d flat space and a $d\!-\!3$ dimensional Euclidean space with finite volume $R^3\times \gamma_{d-3}$.

\subsection{Implications of smoothness of $T_{\mu\nu}$ and Einstein equations} 
\begin{lemma}\label{lemma1}
 Analyticity and finiteness of $T_{\mu\nu}$ at the  horizon in the Gaussian null coordinates imply $T_{rr}=T_{ra}=0$.
\end{lemma}
\begin{proof}
In order $r=0$ to be a smooth surface, $T_{\mu\nu}$ components in the Gaussian null coordinates we have adopted should remain finite and analytic at the horizon (at $r=0$).  This implies that in the near-horizon limit \eqref{rigidscaling}  components of the energy momentum tensor should scale as
$$
T_{rr}\propto \lambda^{-2}\,,\qquad T_{ra}\propto \lambda^{-1}\,,
$$
and all the other components are proportional to $\lambda^0$ or positive powers of $\lambda$. The above in particular implies that  $T_{rr}=T(x)/r^2,\ T_{ra}=T_a(x)/r$ in the near-horizon limit. Requiring that at $r=0$ components of $T_{\mu\nu}$ remain finite at $r=0$ we learn that $T(x)= T_a(x)=0$, or in other words, $T_{rr},\ T_{ra}$ are both vanishing.
\end{proof}
 The above result is of course compatible with the explicit form of the energy-momentum tensors of Maxwell-Dilaton or Maxwell-Scalar theory given in \eqref{T-Lambda}, \eqref{Dilaton-T} and \eqref{gauge-field-T}.
{\begin{lemma}\label{lemma2}
 Einstein equations for the NHEVH ansatz, regardless of the details of the matter fields, imply $ g_a=0$ and $f_a=G^{-1}\partial_a G\equiv 2 \partial_a K$. Therefore, NHEVH metric reduces to
\be
ds^{2} = e^{-2K} \left[\rho^{2}(-\tilde{F} dv^{2}+ 2\tilde{H} dv d\phi+ d\phi ^{2}) +2 dv d\rho \right]   +\gamma_{ab} dx^{a}dx^{b}.
\ee
where $\tilde{H}=e^{-2K}\ {H}$, $\tilde{F}=e^{-2K}\ {F}$ and $\rho=r e^{2K}$ .
\end{lemma}}
\begin{proof}
 From Lemma \ref{lemma1} and Einstein equations \eqref{Ein-Eq} and that in the Gaussian null coordinates $g_{rr}=0$ we learn that
\be
G_{rr}=0\,,\qquad G_{ra}=0\,.
\ee
Moreover, from $g_{rr}=0$, vanishing of $G_{rr}$ means  $R_{rr}=0$ and
{
\be
R_{rr}=\frac{1}{2r^2 \det(g)}\left[r^2G \det(\gamma)+\det(g)\right]\,.
\ee
On the other hand, we know that $\det(g)=-r^2\det(\gamma)\left[G-g_ag_b \gamma^{ab
} \right]$. Then $R_{rr}$ is simply written as
\bea
R_{rr}=\left(\frac{\det(\gamma)}{2\det(g)}\right) g_a g_b\gamma^{ab}\,.
\eea
Since $\gamma_{ab}$ is a positive definite metric, }vanishing $R_{rr}$ yields $g_a=0$.

Noting that in the Gaussian null coordinates $g_{ra}$ components of metric are zero, vanishing of $ra$ components of the Einstein tensor implies vanishing of similar components of the Ricci tensor, $R_{ra}=0$. One may readily compute these components for the NHEVH ansatz \eqref{NHEVH} at $g_a=0$:
{
\be
R_{ra}=\frac{G f_a-\partial_a G}{2rG}\,.
\ee
}
$R_{ra}=0$ then yields
\be
f_a=G^{-1}\partial_a G \equiv 2 \partial_a K\,.
\ee
We note that $G$ is a positive definite function and the above choice for $K$ (and its real valued-ness) is made to enforce this fact.
Given the above form for $f_a$, one can remove the $vx^a$ components of metric by the following redefinition of the $r$ coordinate
\be
r\to \rho=r e^{2K}\,.
\ee
Therefore, regardless of the matter content of the theory and only based on the existence of EVH black hole solution and smoothness of $T_{\mu\nu}$ components at the horizon, the most general NHEVH ansatz takes the form
\be\label{NHEVH-simplified}
ds^{2} = e^{-2K} \left[\rho^{2}(-\tilde{F} dv^{2}+ 2\tilde{H} dv d\phi+ d\phi ^{2}) +2 dv d\rho \right]   +\gamma_{ab} dx^{a}dx^{b}\;,
\ee
where $\tilde{H}=e^{-2K}\ {H}$ and  $\tilde{F}=e^{-2K}\ {F}$.
\end{proof}
{\begin{lemma}\label{lemma3}
For theories with $T_{\phi a}=T_{v a}=0$, $\tilde{H}$ and $\tilde{F}$ are constants.
\end{lemma}}
\begin{proof}
Given the above NHEVH ansatz \eqref{NHEVH-simplified}, we make the following two observations:
\begin{enumerate}
\item[A.] For class of theories with $T_{v a}=0$,  $R_{v a}=0$, and
\be
R_{va}=-\frac{3}{2}\; \rho \;( \tilde{H} \partial_a \tilde{H} +\partial_a \tilde{F})=0. 
 \ee
\item[B.] If $T_{\phi a}$ components of energy momentum tensor vanish, then we learn that $R_{\phi a}=0$. With the metric \eqref{NHEVH-simplified}
\be
R_{\phi a} =\frac{3}{2} \rho\;  \partial_a \tilde{H} =0 \quad \Longrightarrow \quad \tilde{H}=const.
\ee
\end{enumerate}
Then $T_{\phi a}=T_{v a}=0$ and Einstein equations imply $\tilde{H}$ and $\tilde{F}$  are constant.
\end{proof}
{We are now ready to note  one of the main results of this paper. }
{\begin{theorem} \label{theorem1}
The near-horizon geometry of any EVH black hole in Einstein gravity which matter fields have a finite and analytic energy momentum tensor $T_{\mu\nu}$ at the horizon and $T_{\phi a}=T_{v a}=0$, is given by the following metric
\bea\label{NHEVH-ansatz-after-AB-conditions}
ds^{2} = e^{-2K}\left[A_0\rho^2 dv^{2} +2 dv d\rho + \rho^2 d\phi^{2}\right]+\gamma
_{ab} dx^{a}dx^{b},
\eea
where $A_{0}$ is a constant and the 3d $(\rho,v,\phi)$ part of the metric is maximally symmetric.
\end{theorem}
\begin{proof}
Using Lemmas \ref{lemma1},\ref{lemma2},\ref{lemma3} and a shift in $\phi\rightarrow \phi-\tilde{H} v$, the near horizon metric of an EVH black hole takes the form \eqref{NHEVH-ansatz-after-AB-conditions} where  $A_{0}=-(\tilde{H}^2+\tilde{F})$.
\end{proof}}
We would like to make some comments about the above metric:
\bi
\item[I.] The 3d $(\rho,v,\phi)$ part of metric   depends  only on a single constant $A_0$.
\item[II.] At constant $x^a$ surfaces, the 3d $(\rho,\ v,\ \phi)$ part of \eqref{NHEVH-ansatz-after-AB-conditions} is a constant curvature, maximally symmetric space with Ricci scalar $R=6A_0$. Therefore, for negative $A_0$ this 3d
part is (locally) AdS$_3$. For the special case of $A_0=0$ this becomes a locally 3d flat space and if $A_0>0$ we have a dS$_3$.
\item[III.] In our original EVH black hole, we did not impose any condition on the topology of the $d-2$ dimensional horizon, while we required that the area of horizon is finite (and is indeed vanishing in the EVH limit). This in particular, means that $\gamma_{d-3}$ is expected to  have finite volume.
\item[IV.] As we know from various explicit examples \cite{EVH-3,EVH-4}, $e^{-2K}$ may have zeros in some isolated points. At these points metric \eqref{NHEVH-ansatz-after-AB-conditions} (and possibly the $\gamma_{ab}$ metric) may have a curvature singularity. Note also that $e^{-2K}$ should remain everywhere finite as is implied by comment III. above.
{\item[V.] One can show that the geometry \eqref{NHEVH-simplified} has three Killing vector fields, whose explicit form for $A_0\neq 0$ cases may be written as\footnote{ $A_0=0$ case  corresponds to the three dimensional cone, defined by hypersurface $\rho^2=x^2+y^2$ in a flat four dimensional space with metric $ds^2=2d\rho dv+ dx^2 +dy^2$, where $\phi$ is angular coordinate in $xy$ plane. The corresponding Killing vectors are $\partial_v,\; \partial_{\phi},\;  \phi \partial_v +\frac{\partial_{\phi}}{\rho},\; v\partial_v-\rho\partial_{\rho}+\phi\partial_{\phi},\; \phi^2\partial_v-2\partial_{\rho}+\frac{2\phi}{\rho}\partial{\phi},\; 2v\phi\partial_v -2\rho \phi\partial_{\rho}+(\phi^2+\frac{2v}{\rho})\partial_{\phi}. $ }
\be\label{K1}
\partial_v\equiv -A_0(L_+-\tilde{L}_+),\quad \partial_\phi \equiv \sqrt{-A_0}\;(L_++\tilde{L}_+) \quad v\partial_v-\rho\partial_\rho+\phi\partial_\phi\equiv L_0+\tilde{L}_0.
\ee}
As discussed, appearance of the extra Killing vector field $L_0+\tilde{L}_0$ in the NHEVH geometry  is only based on the existence of the near horizon limit and is independent of imposing equations of motion or details of the theory and its dimension. This feature may be compared with a similar fact for the near-horizon extremal geometries \cite{KL-review}, where in the near-horizon geometry and before imposing the equations of motion, the isometry from just $\partial_v$ is enhanced to $\partial_v,\ v\partial_v-r\partial_r$.
\item[VI.]
As we will demonstrate in section \ref{NHEVH-Matter-field-section},  not only the metric but also the dilaton \eqref{Phi-ansatz} and gauge field configurations \eqref{gauge-field-ansatz}  are invariant under the above Killing vectors.
\item[VII.] As discussed,  if the two conditions $R_{va}=R_{\phi a}=0$ are also met, we get three more Killing vectors, whose explicit form depends on $A_0$.{ Together with (\ref{K1}), we have following six Killing vectors (for $A_0\neq0$ cases)
\be
L_+=-\frac{1}{2A_0}\partial_v+\frac{1}{2\sqrt{-A_0}}\partial_{\phi},\quad \tilde{L}_+=\frac{1}{2A_0}\partial_v+\frac{1}{2\sqrt{-A_0}}\partial_{\phi} \nonumber
\ee
\bea
&& L_0=\frac{1}{2}\left(v+\frac{\phi}{\sqrt{-A_0}}\right) \partial_v- \frac{1}{2}\rho\; \partial_\rho +\frac{1}{2}\left({\phi} +\sqrt{-A_0}v+\frac{1}{\rho\sqrt{-A_0}}\right)\partial_{\phi},\nonumber\\
&& \tilde{L}_0=\frac{1}{2}\left(v-\frac{\phi}{\sqrt{-A_0}}\right) \partial_v- \frac{1}{2}\rho\; \partial_\rho +\frac{1}{2}\left({\phi} -\sqrt{-A_0}v-\frac{1}{\rho\sqrt{-A_0}}\right)\partial_{\phi} \nonumber
\eea
\bea
&& L_-=\frac{1}{2}\left(\phi+{v}\sqrt{-A_0}\right)^2 \partial_v- \left(1-A_0\rho v +\sqrt{-A_0}\rho\phi\right) \partial_\rho \nonumber\\
&&\hspace{8mm}+\frac{1}{2\rho}\left(\sqrt{-A_0} v+\phi\right)\left(2+\sqrt{{-}A_0}\rho \phi - A_0 v\rho\right)\partial_{\phi},
\eea
 \bea
&& \tilde{L}_-=-\frac{1}{2}\left(\phi-{v}\sqrt{-A_0}\right)^2 \partial_v{+} \left(1-A_0\rho v -\sqrt{-A_0}\rho\phi\right) \partial_\rho \nonumber\\
&&\hspace{8mm}+\frac{1}{2\rho}\left(\sqrt{-A_0} v-\phi\right)\left(2-\sqrt{{-}A_0}\rho \phi - A_0 v\rho\right)\partial_{\phi},
\eea
The above six Killing vectors form the following algebra
\be
[L_i,L_j]=(i-j)L_{i+j},\qquad [\tilde{L}_i,\tilde{L}_j]=(i-j)\tilde{L}_{i+j}.
\ee
The isometry algebra is hence isomorphic to  $so(2,2)\simeq sl(2,R) \times sl(2,R)$ algebra for $A_0<0$, $so(3,1)$  algebra  for  $A_0>0$ and $iso(2,1)$  algebra for $A_0=0$.}

The above enhancement of symmetries, from the two $\partial_v$ and $\partial_\phi$ Killings of the original EVH black hole to the above six Killings is analogous to the situation in the extremal stationary black holes, where the timelike Killing vector is enhanced to three Killing vectors upon imposing (a part of) equations of motion \cite{KL-review}.
\item[VIII.] It is well known that dealing with nonlinear differential equations, the symmetries of the source do not necessarily carry over to the metric solution in GR. Nonetheless, the converse is not true: symmetries of metric are necessarily symmetries of the source. Therefore, if the two conditions $R_{va}=R_{\phi a}=0$ are met, the metric has six Killing vectors discussed above and hence, as we will explicitly see, the dilaton and gauge fields also exhibit the same symmetries.
\item[IX.] One may compute the Ricci curvature of metric \eqref{NHEVH-ansatz-after-AB-conditions} {using a null-orthonormal frame for the near-horizon metric $\mathrm{e^A}$, where $A=+,-,\phi, a $ and
\be
\mathrm{e}^+=e^{-K} dv,\quad  \mathrm{e}^-=e^{-K} (d\rho+\frac{1}{2}A_0\rho^2 dv ),\quad  \mathrm{e}^{\phi}=e^{-K}\rho d\phi,\quad
\mathrm{e}^a=\hat{\mathrm{e}}^a\ee
where $\hat{\mathrm{e}}^a$ are vielbeins for the horizon metric $\gamma_{ab}$ and the entire space-time metric is $g=\eta_{AB}\mathrm{e}^A\mathrm{e}^B=2\mathrm{e}^+\mathrm{e}^-+{e^{\phi}e^{\phi}}+\mathrm{e}^a \mathrm{e}^a $. The connection 1-forms which are defined by $d\mathrm{e}^A+\omega^{A}_{\;\;B}\wedge \mathrm{e}^B=0$ turn out to be
\bea
&& \omega_{+-}=A_0\rho e^K \mathrm{e}^+,\quad \omega_{+\phi}=A_0\rho e^K \mathrm{e}^{\phi},\qquad \omega_{+a}=-\nabla_a K \;\mathrm{e}^- \nonumber \\
&& \omega_{- \phi}=-\frac{e^K}{\rho}\mathrm{e}^{\phi},\quad\quad \omega_{-a}=-\nabla_aK\; \mathrm{e}^+,\quad\;\; \omega_{\phi a}=-\nabla_a K\; \mathrm{e}^{\phi},
\eea
and $\omega_{ab}=\hat{\omega}_{ab}$. The curvature two-form  and  the Riemann
tensor are defined as $\Omega_{AB} = d\omega_{AB}+\omega_{AC}\wedge \omega^C_{\;\;B}$  and $\Omega_{AB}=\frac{1}{2}R_{ABCD} \; {\mathrm e}^C\wedge {\mathrm e}^D$.  This can be used to evaluate Ricci tensor components:}
\be\label{Ricci-ab}
R_{ab}=\hat{R}_{ab}-3 \nabla_a K \nabla_{b}K +3 \nabla_{a}\nabla_b K,
\ee
where $\hat R_{ab}$ is the Ricci curvature of metric $\gamma_{ab}$ and
\be
\label{Ricci-3d}
R_{\mu\nu}= \left[e^{-2K}\left(\nabla^2K-3 (\nabla K)^2\right) +2 A_0 \right] \tilde{g}_{\mu\nu}
\ee

where $\tilde{g}$ is metric on the 3d part ($A_0\rho^2 dv^{2} +2 dv d\rho + \rho^2 d\phi^{2}$). The scalar curvature is given by
\be\label{Ricci-scalar}
R=\hat{R} +6A_0 e^{2K} + 6 \nabla^2K -12(\nabla K)^2
\ee
where $\hat{R}$ is the scalar curvature of ${\gamma}_{ab}$.

\ei

\subsection{Implications of strong energy condition for  NHEVH solution}\label{SEC-subsection}
We showed in the previous section, based on smoothness of $T_{\mu\nu}$ and assuming vanishing of some components of energy-momentum tensor, that the near-horizon of an EVH black hole has a 3d maximally symmetric subspace. Our previous analysis, however, did not specify whether this 3d part is AdS$_3$, $R^3$ or dS$_3$. In the following theorem we state implications of
Strong Energy Condition (SEC) on the curvature of the 3d part.
\begin{theorem}\label{theorem2}
Strong energy condition implies the 3d part of near-horizon of an EVH black hole with non-positive cosmological constant  $\Lambda \leq 0$ is either AdS$_3$ or flat. The flat case can only occur for $\Lambda=0$ and the geometry is a direct product of $R^3$ and a $d\!-\!3$ dimensional space of finite volume.
\end{theorem}
\begin{proof}
We start with a generic theory independent analysis of the equations of motion. To this end, let us recall that $d$ dimensional Einstein equations \eqref{Ein-Eq} may be written as
\be\label{Ricci-eq}
R_{\mu\nu}-\frac{2\Lambda}{d-2}g_{\mu \nu}=T_{\mu\nu}-\frac{1}{d-2} T g_{\mu\nu}\,.
\ee
Next, we note the SEC implies
\be\label{SEC}
(T_{\mu\nu}-\frac{1}{d-2} T g_{\mu\nu})t^\mu t^\nu \geq 0
\ee
for any (future-oriented) time-like vector field $t^\mu$.

Assuming SEC, we then need $(R_{\mu\nu} -\frac{2\Lambda}{d-2}g_{\mu \nu})t^\mu t^\nu\geq 0$. Recalling \eqref{Ricci-3d},  SEC implies
\footnote{Note that \eqref{SEC-I} is a necessary, but not necessarily sufficient, condition for SEC.}
\be\label{SEC-I}
\nabla^2K-3(\nabla K)^2 -\frac{2\Lambda}{d-2}+2A_0 e^{2K} \leq 0\,.
\ee
Now, consider the integral below
\[
\int_{\gamma_{d-3}} d^{d-3}x \sqrt{\det\gamma}\ e^{-\alpha K} \left(\nabla^2K-3(\nabla K)^2-\frac{2\Lambda}{d-2} +2A_0e^{2K}\right) \leq 0
\]
for an arbitrary real, $\alpha\geq 3$. Integrating by-part we obtain \footnote{Here we have used the fact that although the $d\!-\!3$ space with  $\gamma_{ab}$ can have curvature singularity, it has a finite volume and in fact $\int \sqrt{\det\gamma} e^{-c K}$ for any $c\geq 0$ is finite and positive. We have also used the point that one may drop the ``surface integral'' $\int_{\gamma_{d-3}} d^{d-3}x \nabla_a(\sqrt{\det\gamma}  \gamma^{ab} \nabla_be^{-\beta K})$ for $\beta\geq 1$. This latter is based on the fact that $\gamma_{d-3}$ is a finite volume space and non-compactness (punctures) it might have is coming at points $e^{-K}$ may vanish. }
\be
\int_{\gamma_{d-3}} d^{d-3}x \sqrt{\det\gamma}\ e^{-\alpha K} \left[(\alpha-3)(\nabla K)^2 -\frac{2\Lambda}{d-2}+2A_0e^{2K}\right] \leq 0
\ee
Therefore, SEC implies
\be  A_0<0,   \quad \text{unless when} \ \nabla_a K=0\ \text{and}\ \Lambda=0,\ \text{where }\  A_0=0\ \text{is also possible}.
\ee
That is, if $\nabla_a K\neq 0$, for $\Lambda\leq 0$, we will always get a space which is a warped product of an AdS$_3$ with $\gamma_{d-3}$. If $\nabla_a K=0$ (i.e. a constant warp factor) then we have the option of getting 3d flat space $A_0=0$, only if $\Lambda\geq 0$. Therefore, within the assumptions of our theorem $A_0=0$ may only be allowed for $\Lambda=0$ case. For generic $\Lambda>0$ the above analysis does not yield a restriction on the sign of $A_0$.
\end{proof}

\section{Near horizon limit and {matter fields}}\label{NHEVH-Matter-field-section} 

So far we did not make any assumption about the matter content of the theory to which our NHEVH geometry is a solution. In this section we study behavior of other fields in the near-horizon EVH limit.  As we saw in the previous section, requiring the metric and other fields (physical observables) at the horizon to be smooth and to admit smooth $r\to 0$ limit, imposes strong conditions on the form of the metric. In this section we explore implications of similar requirements on the other  fields in the problem.

We consider Einstein-Maxwell-scalar-$\Lambda$ theory in generic dimension. The scalar sector may be dilaton fields (with a shift symmetry) or   scalars with a potential. The gauge field part in odd dimensions may also include a Chern-Simons term. Our analysis is hence quite generic and includes the bosonic part of all gauged or ungauged supergravities with $U(1)$ gauge symmetry.   The black holes we consider are generically  solutions to Einstein-Maxwell-Dilaton-$\Lambda$ (EMD-$\Lambda$)  theory or the gauged supergravity theories, where we have some Maxwell fields coupled to scalars with potential terms. Let us consider the generic  action of the form
\be\label{action-generic}
\mathcal{\mathbf{L}}=\frac{1}{2} (R-2\Lambda)-\frac12 g^{\mu\nu}\mathcal{G}_{IJ}(\Phi)\partial_\mu\Phi^I\partial_\nu\Phi^J-V(\Phi^I)-\mathrm{e}^{c^p_I\Phi_I} F_{\mu\nu}^{p}F_{\alpha\beta}^{p} g^{\mu\alpha}g^{\nu\beta}+ \mathcal{L}_{CS}
\ee
where $\mathcal{G}_{IJ}$ is metric on the space of dilaton/scalar  fields (which is taken to be positive definite), $V(\Phi)$ is the potential for scalar fields, $F^p_{\mu\nu}=\partial_{[\mu}A_{\nu]}^p$ denotes the field strength of gauge fields $A^p$, and $\mathcal{L}_{CS}$ denotes a possible Chern-Simons term which may exist in odd dimensions.

Let us start with the contribution of the cosmological constant term to the energy-momentum tensor:
\be\label{T-Lambda}
T^\Lambda_{\mu\nu}=-\Lambda g_{\mu\nu}\,,
\ee
and hence for the metric ansatz \eqref{NHEVH}, $T^\Lambda_{rr}=T^\Lambda_{ra}=0$.

\subsection{Dilaton/scalar fields}

An EVH black hole solution to \eqref{action-generic} involves scalar fields where $\Phi_I=\Phi_I(r;x)$,
where we already used invariance under $\partial_v,\ \partial_\phi$ diffeomorphisms, ${\cal L}_v\Phi_I= {\cal L}_\phi\Phi_I=0$. Upon taking the near horizon limit \eqref{rigidscaling}, we learn that for NHEVH solution ansatz
\be\label{Phi-ansatz}
\Phi_I=\Phi_I(x^a)\ .
\ee
It is clearly seen that the above form for the scalar fields, being independent of $r,v,\phi$, is explicitly invariant under the six Killing isometries of the 3d part.

The equation of motion and the energy-momentum tensor for the scalar fields are
\bea
\frac{1}{\sqrt{-{g}}}\partial_\mu(\sqrt{-{g}}g^{\mu\nu} \mathcal{G}_{IJ}(\Phi)\partial_\nu\Phi^J)&=&\frac12 g^{\mu\nu}\frac{\partial\mathcal{G}_{JK}}{\partial\Phi_I} \partial_\mu\Phi^J\partial_\nu\Phi^K+\frac{\partial V}{\partial\Phi_I}+c^a_I\mathrm{e}^{c^a_I\Phi_I} F_{\mu\nu}^{a}F_{\mu\nu}^{a} \label{Dilaton-eom} \\
T_{\mu\nu}^{\Phi}=\mathcal{G}_{IJ}(\Phi)\partial_\mu\Phi^I\partial_\nu\Phi^J\!\!&-&\!\!\frac12 g_{\mu\nu}\ \mathcal{G}_{IJ}(\Phi)g^{\alpha\beta}\partial_\alpha\Phi^I\partial_\beta\Phi^J- V(\Phi) g_{\mu\nu} \label{Dilaton-T}
\eea
One can readily see that with \eqref{Phi-ansatz}, the $T^\Phi_{\mu\nu}$ components with $\mu$ or $\nu=r,v,\phi$ are proportional to the metric $g_{\mu\nu}$. In particular, one can also see that
\be\label{T-Phi-compts}
T^\Phi_{rr}=0\,,\quad T^\Phi_{ra}=0\,,\qquad T^\Phi_{va}\propto f_a\,,\quad T^\Phi_{\phi a}\propto g_a,
\ee
where $g_a$ and $f_a$ are respectively related to the off-diagonal $g_{\phi a}$ and $g_{v a}$ components of the metric \eqref{NHEVH}. Therefore, a generic scalar theory in the near-horizon EVH geometry satisfies the conditions of Lemmas 1,2 and 3, and hence Theorem 1 of previous section.

{\subsection{Gauge fields}}

Requiring the gauge field one-form $A^p$, (1) to be $\partial_v,\ \partial_\phi$ invariant, i.e. ${\cal L}_v A^p= {\cal L}_\phi A^p=0$ and, (2)  to have a well-defined near horizon limit \eqref{rigidscaling}, one learns that 
the most general near-horizon gauge field ansatz takes the form
\be\label{gauge-field-ansatz}
A^p= r e^p(x) dv+\frac{1}{r} h^p(x) dr+ r b^p(x) d\phi+ A_a^p(x) dx^a \,.
\ee
The gauge field energy-momentum tensor is
\bea
T_{\mu\nu}^{A}=2\mathrm{e}^{c^p_I\Phi_I} \left(F_{\mu\alpha}^{p}F^p_{\nu\beta} g^{\alpha\beta}-\frac14  g_{\mu\nu} (F^p_{\alpha\beta})^2\right).\label{gauge-field-T}
\eea
Requiring smoothness of $T^A_{\mu\nu}$ components, we can restrict the gauge field ansatz more. In particular, using \eqref{gauge-field-ansatz} and metric given by \eqref{NHEVH}, we have
\bea
T_{rr}^A= (b^p)^2 ~g^{\phi \phi}+2b^p\frac{\partial_a h^p}{r}~g^{\phi a}+\frac{\partial_a h^p\partial_b h^p}{r^2}~g^{ab} \,,
\eea
which is the norm of the spacelike  vector $(0,0,b^p,\frac{\partial_a h^p}{r})$. As implied by Lemma 1, $T^A_{rr}$ should vanish at $r=0$ and hence $b^p(x)=0$ and $\partial_a h^p=0$.\footnote{Note that $g^{\phi \phi} \propto 1/r^2$ and $g^{\phi a} \propto 1/r$.} Having these two conditions, one can readily see that $T_{ra}^A$ vanishes too. Therefore, the Lemma 1 of the previous section is explicitly verified for Maxwell gauge fields.
Gauge field ansatz, up to possible gauge transformations, becomes
\bea\label{gauge-field-ansatz-2}
A^p= r e^p(x) dv+ A_a^p(x) dx^a\,.
\eea

The gauge field equations of motion are
\bea
\frac{1}{\sqrt{-{g}}}\partial_\mu(\sqrt{-{g}}g^{\mu\alpha}\ \mathrm{e}^{c^p_I\Phi_I} F_{\alpha\nu}^{p})+ \frac{\delta\mathcal{L}_{CS}}{\delta A_\nu^p}=0.\label{gauge-field-eom}
\eea
$\nu=r$ component of above equations leads to $e^p(x)=0$.
 After setting $e^p$ to zero in \eqref{gauge-field-ansatz-2}, one finds
\footnote{Note that in the Gaussian null coordinates \eqref{NHEVH} $rr$ and $ra$ components of metric are zero.}
\be\label{T-A-compts}
T^A_{rr}=0\,,\quad T^A_{ra}=0\,,\qquad T^A_{va}\propto f_a\,,\quad T^A_{\phi a}\propto g_a,
\ee
where $g_a$ and $f_a$ are respectively related to the off-diagonal $g_{\phi a}$ and $g_{v a}$ components of the metric \eqref{NHEVH}. Therefore, the Maxwell gauge fields in the near-horizon EVH limit also satisfy the conditions of Lemmas 1,2,3 and hence Theorem 1. follows for them.

It is worth noting that we did not employ  scalar field equation of motion to show it is independent of $r,v$ and $\phi$ but for gauge field, equations of motion are needed. Instead of equations of motion, however, we could also use invariance of  gauge fields under Killing vectors of metric \eqref{NHEVH}. All in all, for the most general NHEVH gauge field ansatz, $e^p(x), b^p(x)$ and $\partial_a h^p$ in \eqref{gauge-field-ansatz} should  vanish and hence we remain with
\be\label{gauge-field-simplified}
A^p= A^p_a(x^b) dx^a,\qquad F^p=F^p_{ab}dx^a\wedge dx^b\,.
\ee
One can readily check that the gauge field given above is compatible with the gauge field equations of motion \eqref{gauge-field-eom}, with the metric \eqref{NHEVH-ansatz-after-AB-conditions} and dilaton \eqref{Phi-ansatz}.

In summary, the most general ansatz for the NHEVH solutions of the theory \eqref{action-generic} in any dimension is given through \eqref{NHEVH-ansatz-after-AB-conditions}, \eqref{Phi-ansatz} and \eqref{gauge-field-simplified}.

\subsection{More on Einstein equations}
To fully specify the solution we need to impose and solve the rest of equations of motion. This involves scalar and gauge field equations of motion \eqref{Dilaton-eom} and \eqref{gauge-field-eom} and the Einstein equations:
\be
\hspace*{-3mm}R_{\mu\nu}=\frac{2}{d-2} \left(\Lambda+ V(\Phi)-\frac12\mathrm{e}^{c^p_I\Phi_I}(F^p)^2\right)g_{\mu\nu}+ \mathcal{G}_{IJ}(\Phi)\partial_\mu\Phi^I\partial_\nu\Phi^J+2\mathrm{e}^{c^p_I\Phi_I}F_{\mu\alpha}^{p}F^p_{\nu\beta} g^{\alpha\beta},
\ee
where we used \eqref{Dilaton-T} and \eqref{gauge-field-T}. For the dilaton and gauge field ansatz given in previous sections, we find
\be
\hspace*{-6mm}R_{\mu\nu}=\left\{\begin{array}{cc}
\hspace*{-65mm}\frac{2}{d-2}[\Lambda+ V(\Phi) -\frac12\mathrm{e}^{c^p_I\Phi_I} (F^p)^2]g_{\mu\nu},\quad & \mu,\nu=\rho,v,\phi\cr \ \ \ &\cr
\frac{2}{d-2}[\Lambda+ V(\Phi) -\frac12\mathrm{e}^{c^p_I\Phi_I} (F^p)^2 ]\gamma_{ab}+ \mathcal{G}_{IJ}(\Phi)\partial_a\Phi^I\partial_b\Phi^J+2\mathrm{e}^{c^p_I\Phi_I}F_{ac}^{p}F^p_{bd}\gamma^{cd},\ \ & \mu,\nu=a,b
\end{array}\right.
\ee
where $(F^p)^2=F_{ac}^{p}F^p_{bd}\gamma^{cd}\gamma^{ab}\geq0$.
As expected and argued in the previous section, the above is of course compatible with the 3d part being a maximally symmetric space. In this work we do not intend to make classification of the NHEVH solutions. This question has been dealt with  in \cite{EVH/CFT} in four dimensions and partially in \cite{SO22} in five dimensions.

We also comment that, as it is readily seen, the theory \eqref{action-generic} satisfies the strong energy condition \eqref{SEC} for theories with $V(\Phi)\leq 0$. This is of course compatible with the well known fact about the ungauged and gauged supergravity theories, that they satisfy the strong energy condition.
For these cases our Theorem 2 of section \ref{SEC-subsection} is applicable; in this class of theories if $V(\Phi)+ \Lambda\leq 0$, the 3d part is either AdS$_3$ or 3d flat space. Next, we explore the 3d flat case (which is a much more restrictive case).

\paragraph{On the existence of  $\nabla_a K=0$ solutions.} The 3d part of equations of motion for this case implies that
\be
\frac12\mathrm{e}^{c^p_I\Phi_I} (F^p)^2=\Lambda+V(\Phi)\,.
\ee
The above has solutions only when $\Lambda+V(\Phi)\geq 0$, and for $\Lambda+V(\Phi)=0$ case we are forced to turn off the gauge fields. The $d-3$ dimensional part of the equations of motion,
\be
R_{ab}=\mathcal{G}_{IJ}(\Phi)\partial_a\Phi^I\partial_b\Phi^J+2\mathrm{e}^{c^p_I\Phi_I}F_{ac}^{p}F^p_{bd}\gamma^{cd},
\ee
then implies the $d\!-\!3$ dimensional part of metric should have a positive definite Ricci scalar. For $\Lambda+V(\Phi)=0$ case, $F_{ab}=0$ and hence the solution may only exist if $R_{ab}=\mathcal{G}_{IJ}(\Phi)\partial_a\Phi^I\partial_b\Phi^J$ and $\nabla^2 \Phi_I=\partial_IV(\Phi)$ have simultaneous solutions.

An interesting case is Einstein vacuum solutions, when gauge field and dilatons are turned off. This implies $\Lambda+V(\Phi)$ should  be necessarily zero and the $d-3$ dimensional part should be Ricci flat. $T_{\mu\nu}=0$ obviously satisfies strong energy condition. For $d=4,5,6$, $d-3$ dimensional  Ricci flatness implies vanishing of the Riemann curvature and hence the whole space is (locally) flat and  there are no nontrivial solution for $d\leq 6$. For $d\geq7$ we have other options, as Ricci flatness in $d-3\geq 4$ does not imply vanishing curvature. For $d\geq 7$ vacuum solutions with $A_0=0$ can hence be classified through Euclidean, compact Ricci flat $d-3$ dimensional geometries. We, however, note that such geometries may not be related to any (EVH) black hole solution in the near horizon limit.

\paragraph{More on $A_0\leq 0$, AdS$_3$ case.} This case is more generic as generic field configurations in our theory satisfy SEC and $\nabla_aK\neq0$. We will review several examples  in this class in section \ref{examples-section}. These solutions generically come from known EVH black hole solutions in the near horizon limit. It is worth noting that, since in our near-horizon limit \eqref{rigidscaling} we also scale $\phi$, if viewed as near-horizon limit of EVH black hole solutions, the AdS$_3$ factors are ``pinching AdS$_3$'' geometries.

\section{Near-horizon near-EVH geometries}\label{near-EVH-section}

So far, using Einstein field equations in the presence of non-positive cosmological constant, we have proved that for an EVH black hole the near-horizon geometry has an AdS$_3$ throat or a 3d flat spacetime. However, in section \ref{NHEVH-ansatz-section} and in particular in \eqref{genericNH}, we discussed the most general form of near-horizon ``near-EVH'' geometry which has several extra functions compared to the EVH case. In this section we show that similar analysis as made in  section \ref{NHEVH-EOM-section}, leads to the fact that in cases with AdS$_3$ throat the near-horizon near-EVH geometry should be the same geometry as in the EVH case, with the AdS$_3$ part replaced with a BTZ black hole. For the cases where the 3d part is (locally) flat we show that the 3d part of the near-EVH geometry corresponds to a particle of a given mass and angular momentum.
{\begin{theorem}\label{theorem3}
The 3d part of near horizon of a near-EVH black hole with non-positive cosmological constant is either a (pinching) BTZ black hole or a rotating massive particle on the flat spacetime.
\end{theorem}}
\begin{proof}
We note that both sides of the Einstein equations \eqref{Ein-Eq} may be expanded in powers of $\alpha\equiv\epsilon/ \lambda$, as in \eqref{double-expansion}. The leading order equations in powers of $\epsilon/\lambda$ does not involve near-EVH parameters and hence they lead to the previous results, i.e.
\bea
g_a=0 , ~ f_a=\frac{\partial_a G}{G}, ~ H=\tilde{H} G\ ~ F=\tilde{F} G\; ,
\eea
where $\tilde F, \tilde H$ are constant. Moreover, we know that $F^{(1)}$ is a non-negative constant recalling the fact that Hawking temperature is a constant over the horizon (\emph{cf}. discussions in section 2). Upon the coordinate transformation $r\to \rho=e^{2K} r$ as before, ($G\equiv e^{2K}$), we find the most general form of the near-EVH metric
\bea\label{NH-Near-EVH-simplified}
ds^{2}\hspace{-.3cm}&=&\hspace{-.3cm}e^{-2K}\hspace{-.1cm}\left[-\rho(\rho\tilde{F}+\alpha F^{(1)}) dv^{2}+ 2\rho(\tilde{H}\rho+\alpha H^{(1)}) dv d\phi+ [(\rho+\alpha R)^2+\alpha^2 J ]d\phi ^{2} +2 dv d\rho \right] \cr
&+&\hspace{-.3cm}2\alpha \ g^{(1)}_a dx^a d\phi  +\gamma_{ab} dx^{a}dx^{b}\;.
\eea
The above metric has $d$ more unknown functions $H^{(1)}, R, J$ and $g^{(1)}_a$ (and a constant $F^{(1)}$) compared to the EVH case \eqref{NHEVH-simplified}. Of course as expected in $\alpha=0$ we recover the NHEVH ansatz \eqref{NHEVH-simplified}.

Since $|\partial_\phi|^2>0$ and that determinant of metric has a definite sign,  $J \geq 0$ and  $(\rho+ \alpha R)^2+\alpha^2 J > \alpha^2\gamma^{ab}g^{(1)}_ag^{(1)}_b$ where $\gamma^{ab}$ is inverse of metric $\gamma_{ab}$.

As in the EVH case, to determine or restrict the other unknown functions in the near-EVH metric, we impose Einstein equations. These equations should be valid for any value of $\alpha$ parameter, as long as $\alpha\lesssim 1$. As argued in section 4, smoothness of energy-momentum tensor in the near-horizon limit implies $T_{rr}=T_{ra}=0$. These conditions remain true for near-EVH case and therefore we have
$R_{rr}=R_{ra}=0$. A direct computation of the Ricci reveals that
\be
R_{rr}=R_{ra}=0 \quad \Rightarrow\quad J=\gamma^{ab}g^{(1)}_ag^{(1)}_b\,.
\ee
With the above, \eqref{NH-Near-EVH-simplified} may be written as
\bea\label{NH-Near-EVH-simplified-II}
ds^{2} &=& e^{-2K} \left[-\rho(\rho\tilde{F}+\alpha F^{(1)}) dv^{2}+ 2\rho(\tilde{H}\rho+\alpha H^{(1)}) dv d\phi+ (\rho+ \alpha R)^2d\phi ^{2} +2 dv d\rho \right] \cr
&+&\gamma_{ab}(dx^{a}+\alpha{g^{(1)}}^a d\phi)(dx^{b}+\alpha{g^{(1)}}^b d\phi)\;,
\eea
where ${g^{(1)}}^a =\gamma^{ab}g^{(1)}_b$.

As in the EVH case,  $T_{va}=0$ implies $R_{va}=0$. Next, we note that in the EVH case, as can be directly seen  from \eqref{NHEVH}, $\partial_\phi$ is a Killing vector which is hypersurface orthogonal on the horizon of the original EVH black hole, i.e. at codimension two constant $v$ and $r=0$ surfaces, $\partial_\phi$ is transverse to the constant $\phi$ surfaces. Hereafter, we \emph{assume} that this property remains in the near-EVH case. This assumption, which is justified through many examples of EVH black holes in five or six dimensions, implies $g^{(1)}_a=0$. With this assumption,  $T_{\phi a}=0$ implies $R_{\phi a}=0$, and therefore,
\be
R_{va}=0,\ R_{\phi a}= 0\quad \Rightarrow \quad \tilde F,  H^{(1)},
{\tilde H,}\ R=\mathrm{const.}
\ee

With the above assumptions, we obtain a metric with five constants and the unknown functions $K$, $\gamma_{ab}$ to be determined by the remaining Einstein equations. One of the five constants may be removed by a coordinate transformation $\phi\to \phi+ c v$, with a constant $c$. We note that the 3d $\rho,v,\phi$ part is not a maximally symmetric space, unless the constants are related in a specific way. Such relations comes from components of the Einstein equations along the 3d part, especially recalling that for the class of theories discussed in section \ref{NHEVH-Matter-field-section} the energy-momentum tensor along the 3d part is proportional to its metric. Explicitly, we get maximally symmetric 3d space if
\be\label{flat-condition}
H^{(1)}=2\tilde H R\,,\qquad F^{(1)}=2\tilde F R.
\ee
With the above, the Ricci curvature of the near-EVH metric equals that of the EVH metric, and is given through \eqref{Ricci-3d} and \eqref{Ricci-ab}. As in the EVH case, if the matter fields satisfy strong energy condition, we deal with two options:
\begin{itemize}
\item If $A_0=-(\tilde F+\tilde H^2)=0$, we have a (locally) flat 3d space. After the shift $\rho\to \rho-\alpha R$ and $\phi\to\phi-\tilde Hv$, and then rescaling $v\to (\alpha \tilde H R)^{-1}v$, $\phi\to (\alpha \tilde H R)^{-1}\phi$ and $\rho\to (\alpha \tilde H R)\rho$ the geometry takes the form
\bea\label{NH-Near-EVH-flat-case}
ds^{2} = e^{-2K} \left[ dv^{2}+ \frac{2}{\tilde{H}}dv d\phi+ \rho^2 d\phi ^{2} +2 dv d\rho \right] +\gamma_{ab}dx^{a}dx^{b}\;,
\eea
where if the $\phi$ direction in the original near-EVH black hole had a $[0,2\pi]$ range, the $\phi$ coordinate in the above metric is ranging over $[0,2\pi \alpha\tilde H R\lambda]$. One can readily see from discussions of previous sections that \eqref{flat-condition} arises from the equations of motion in the EMD-$\Lambda$ theory. In fact, with \eqref{flat-condition}, the 3d part of metric  is locally flat and denotes a particle of a given mass and angular momentum proportional to $\tilde H$ \cite{DJT}.

\item If $A_0=-(\tilde F+\tilde H^2)<0$, then we have a locally AdS$_3$ space, with metric
\bea\label{NH-Near-EVH-AdS-case}
\hspace*{-16mm}ds^{2} = e^{-2K}\!\!\left[-\tilde F\rho (\rho+2\alpha  R) dv^{2}+ 2\tilde{H} \rho(\rho+2\alpha  R) dv d\phi+ (\rho+ \alpha R)^2d\phi ^{2} +2 dv d\rho \right]\!\!
+\gamma_{ab}dx^{a}dx^{b}.
\eea
The above denotes a (pinching) BTZ, recalling that $\phi\in [0,2\pi\lambda]$, with inner and outer horizon radii $r_\pm$ and AdS$_3$ radius $\ell$ (\emph{cf.} \eqref{BTZ-NGC})
\be
\ell^2=-\frac{1}{A_0}\,,\qquad r_+=\alpha R\,,\qquad \tilde H=\frac{r_-}{\ell r_+}\,.
\ee
\end{itemize}
\end{proof}
To summarize, the near-horizon near-EVH geometries \eqref{NH-Near-EVH-flat-case} and \eqref{NH-Near-EVH-AdS-case} are solutions to the same theories as the NHEVH geometry \eqref{NHEVH-ansatz-after-AB-conditions}. One may then relate the mass and angular momentum of the 3d part in either of the above geometries to the mass and angular momentum perturbations of the near-EVH black hole from the EVH point (before taking the near-horizon limit). In other words, if one views the near-EVH black holes as excitations around EVH black hole, then the information about these near-EVH black holes  appears as mass and angular momentum of the corresponding 3d geometries  after taking the near horizon limit. This point has been demonstrated through several examples of EVH black holes, some of which may be found in our reference list and needs  further study to which we hope to return in upcoming works.

\section{Examples of EVH black holes and their near horizon geometries}\label{examples-section}

There are several examples of EVH black holes in different dimensions.  In four dimensions it has been shown that EMD theory  admits EVH black hole solutions and their near horizon geometry always has an AdS$_3$ geometry  \cite{EVH/CFT}, in accord with our general theorems in this work. The effect of higher dimensional correction has been studied in \cite{Hossein-HD}.  Several examples of four and higher dimensional EVH black holes have been studied and  their near horizon geometry  analyzed and in all these examples appearance of (pinching) AdS$_3$ or (pinching) BTZ for the near-EVH case has been ubiquitous \cite{EVH-2,AdS4-EVH,Japanese-EVH, EVH-3, EVH-4, EVH-Ring, Hossein-4d, Hossein-5d}. In this section we briefly review some of these known solutions.

\subsection{EVH-BTZ black hole solutions in three-dimensions}
 Let us start with three dimensional Einstein theory with negative cosmological constant and the black hole solution there, the BTZ black hole \cite{BTZ}. To take the EVH limit, we first  write BTZ solution in the Gaussian null coordinate system:
 \be\label{BTZ-NGC}
 ds^2=-\frac{r(r+2r_+)(r_+^2-r_-^2)}{\ell^2 r_+^2} dv^2+2drdv-\frac{2r (r r_-+2r_+ r_-)}{\ell r_+}dv d\varphi +(r+r_+)^2d\varphi^2\;.
 \ee
The inner and outer horizons are located at  $r=-(r_+-r_-)$ and $r=0$, respectively. Hawking temperature and the entropy are
 \be
 T_{H}=\frac{r_+^2-r_-^2}{2\pi r_+ \ell^2},\qquad  S_{BH}=\frac{2\pi r_+}{4 G_3}
 \ee
where $G_3$ is Newton coupling constant in three dimensions and cosmological constant is $\Lambda =-6 \ell^{-2}$. If we define $r_{\pm}=\rho_{\pm} \epsilon$ and take $\epsilon\to 0$ limit, BTZ entropy and temperature vanish while their ratio remains finite and hence in this limit we are dealing with an EVH BTZ black hole \cite{BTZ-EVH}. The near-horizon limit involves scaling radial coordinate as $r=\lambda \rho$ and taking $\lambda\to 0$. The near-horizon EVH limit is obtained when  $\epsilon\ll \lambda \ll 1$:
\be
\label{NHEVH3d}
ds^2=-\left( \frac{\rho_+^2-\rho_-^2}{\ell^2\rho_+^2}\right) \rho^2 d\nu^2 +2d \rho d\nu-\frac{2\rho_-}{\ell\rho_+} \rho^2 d\nu d\chi+ \rho^2 d\chi^2,\quad v=\frac{\nu}{\lambda},\; \varphi=\frac{\chi}{\lambda},
\ee
which is a pinching AdS$_3$ with radius $\ell$.  The near horizon near-EVH limit is obtained when  $\lambda\simeq \epsilon$:
\be
\label{NHnearEVH3d}
ds^2=-\frac{\rho(\rho+2\rho_+)(\rho_+^2-\rho_-^2)}{\ell^2 \rho_+^2}d\nu^2+2d\rho d\nu -\frac{2\rho(\rho\rho_-+2\rho_+\rho_-)}{\ell \rho_+}d\nu d\chi +(\rho+\rho_+)^2 d\chi^2\;,
\ee
which is (pinching) BTZ black hole solution, written in Gaussian null coordinate system. We note that $\chi$ coordinate is pinching i.e $\chi \in [0,2\pi\lambda]$.
\subsection{Four  dimensional EVH black hole in EMD theory}
Near horizon structure of generic EVH black hole solution in EMD-theory in four dimensions has been studied in \cite{EVH/CFT},  for 4d heterotic theories in  \cite{Hossein-4d, Hossein-5d} and for U(1)$^4$ gauged supergravity theory in \cite{AdS4-EVH}.  It has been shown that any EVH black hole in this theory, if it exists, has the following near horizon geometry
\be\label{4d-EVH}
ds^2=R^2 |\sin\theta| \left(ds_{{\mathrm AdS}_3}^2+\frac{1}{4}d\theta^2\right)
\ee
where $R$ is a constant determined by charges carrying by the black hole and $ds_{{\mathrm AdS}_3}^2$ is a metric on the {\it pinching} AdS$_3$
space which can be written as (\ref{NHEVH3d}) in Gaussian null coordinate system. We note that there is no vacuum solution of the form  \eqref{4d-EVH} and one should have specific profile of the dilaton fields \cite{EVH/CFT,SO22}. It has also been  shown that this three dimensional part of the metric is replaced by {\it pinching} BTZ black hole solution
(\ref{NHnearEVH3d}) when we take near-horizon limit of a near-EVH black hole solution. These are all in accord with our general discussions of previous sections.
\subsection{Five dimensional EVH black holes }
Five dimensions is the lowest dimension which we can find EVH black hole in the vacuum Einstein theory. It has been shown in \cite{Bar-Horo, EVH-3} that EVH conditions are met for single spin extremal Myers-Perry black holes \cite{MP} or
for single spin extremal black rings \cite{EVH-Ring}. In the near-horizon limit the EVH hole and ring become identical and have the metric \cite{EVH-Ring}
\be
ds^2=R^2\cos^2\theta  ds_{{\mathrm AdS}_3}^2 + R^2 (\cos^2\theta^2+\tan^2\theta d\psi^2)
\ee
where $R$ is proportional to the (non-zero) angular momentum of the original black hole or black ring solution.

It is also worth mentioning the EVH black hole solutions in U(1)$^3$ five dimensional gauged supergravity
which is studied in \cite{EVH-3, EVH-4}. Their near horizon geometry is given by \cite{SO22}
\be
ds^2=H_{\theta}\Bigg[ \mathcal{R}^2   ds_3^2 +\frac{ a^2 }{\Delta_{\theta}} (d\theta ^2+  \frac{H_0^3}{H_{\theta}^3}\frac{\Delta_{\theta}^2}{\Delta_0^2}\; \sin^2\theta\cos^2\theta d\psi^2)\bigg],
\ee
with
\bea
\Delta_{\theta}=(1-\frac{a^2}{\ell^2}\cos^2{\theta})\,,\quad H_{i}={\cos^2{\theta}+s_i^2},\quad H_{\theta}=H_1^{\frac{1}{3}} H_2^{\frac{1}{3}} H_3^{\frac{1}{3}} ,\quad
\eea
where  constants $p_i$ and $R^2$ are related to $a$ and $s_i$ as
\bea
\mathcal{R}^2= \frac{a^2}{1+\frac{a^2}{\ell^2}(s_1^2+s_2^2+s_3^2+1)}, \quad p_i^2=\frac{2 a^4 s_i^2(s_i^2+1)(1+ \frac{a^2}{\ell^2}s_i^2)}{[1+\frac{a^2}{\ell^2}(1+s_1^2+s_2^2+s_3^2)]^3}.
\eea

In \cite{SO22} a classification of four and five dimensional solutions with local $SO(2,2)$ isometry was provided. These solutions, given our theorems in this work, would hence yield a classification of  the NHEVH solutions in these dimensions.

\section{Discussion and outlook}\label{discussion-section}
In this work we analyzed a particular class of solutions to generic Einstein gravity theories in diverse dimensions. We proved two theorems stating that near horizon limit of Extremal Vanishing Horizon (EVH) black holes generically contain an AdS$_3$ throat. Although our general theorems also allow for having a 3d flat space part in the geometry, in the previously studied examples of EVH black holes we do not have such solutions. In our analysis we assumed that the original EVH black hole solution exists but did not fully specified that solution. In our analysis (\emph{cf.} section \ref{EVH-ansatz-section}) we only used very general properties of such solutions.  It is likely that not having 3d flat cases in explicit examples is another general feature which bears upon more details of the theory in question. It would be interesting to explore this direction.

In our analysis in sections \ref{EVH-ansatz-section}, \ref{NHEVH-ansatz-section} we introduced two parameters $\epsilon, \lambda$; $\epsilon$ parameterizes ``near-EVH-ness'' while $\lambda$ measures how close to the horizon we are.
We only considered cases where we remained ``close to EVH'' while taking the near-horizon $\lambda\to 0$ limit and excluded $\epsilon/\lambda\gg 1$. In our analysis, imposing the EVH condition \eqref{EVH-def-1}, we tied together the way vanishing temperature (extremality) and vanishing horizon area happen, and measured both with the same parameter $\epsilon$. One may try to perform this analysis in a more general setup by disentangling these two; e.g. one may take $T\sim \epsilon\to 0$ while $A_h\sim \epsilon^k$ with $k\geq 0$. In this way we can distinguish three more cases: $k=0$ (and in fact for more general $0\leq k< 1$ case), we just recover the usual extremal black holes; for $k=1$ we obtain our EVH analysis and for $k>1$ we encounter new cases. For each of these cases, we still have the option to choose how $\epsilon$ and the near-horizon parameter $\lambda$ scale with respect to each other.
If we choose $0\leq k<1$ and $\epsilon/\lambda\ll 1$, it is straightforward to see that one  recovers  the usual results of near-horizon extremal geometries \cite{KL-review} and for the same values of $k$, if $\epsilon/\lambda\sim 1$, the AdS$_2$ factor in the near-horizon extremal geometry is replaced by an ``AdS$_2$ black hole'' (i.e. a geometry with metric $-(r^2-r_0^2)dt^2+\frac{dr^2}{r^2-r_0^2}$). Among interesting questions which may then arise in this setup is the orders of limits issue. For example,  one may start with a given $k<1$, take the near horizon limit and then take $k\to 1$ limit. That is, we first take the near-horizon extremal limit and then take the EVH limit.  It can be  checked (following our general analysis of sections \ref{EVH-ansatz-section}, \ref{NHEVH-ansatz-section}, or working through specific examples \cite{BTZ-EVH,EVH-3}) that the resulting geometry is not the same as what we have for the EVH or near-EVH case. In particular, we do not see a locally AdS$_3$ throat. Understanding this more general case with parameter $k$ involved (including $k>1$ case) is an interesting direction for future studies.

Given the generality of the theorems we proved, one may also push for  classification and uniqueness theorems  for the near-horizon EVH solutions. Some first steps in this direction has been taken in \cite{EVH/CFT,SO22}. This is another question we hope tackle in our future studies.

In this work we focused on addressing some issues regarding EVH black holes and their near horizon limit from in a GR solution viewpoint. As a completion of the current study one should also analyze various charges associated with different EVH solutions and explore the counterparts of Smarr relation or first law for this class of solutions. Similar steps for the extremal black holes were taken in \cite{First-law,NHEG-mechanics}.

In four dimensional EVH cases it was shown that the near horizon limit is indeed a decoupling limit \cite{EVH/CFT} and appearance of AdS$_3$ throat motivated proposing an EVH/CFT proposal, according which the near horizon EVH black hole is dual (in the sense of AdS/CFT) to a 2d CFT at the Brown-Henneaux central charge. As first  check for the EVH/CFT one may relate the near-EVH excitations, which appear as BTZ black holes in the near horizon limit, to thermal states in the dual CFT. It is desirable explore the EVH/CFT for specific examples and in particular understand the ``pinching'' feature. Some proposals for the latter was presented in \cite{BTZ-EVH}. As another technically important step in this direction is related to the near-EVH analysis of section \ref{near-EVH-section}. There we assumed the off-diagonal components of metric $g^{(1)}_a$ to be zero; their vanishing, while compatible with all equations of motion and in accord with various known EVH examples, did not directly come out of the equations of motion. In order to establish near horizon limit as a decoupling limit for the near-EVH case it is crucial to show $g^{(1)}_a=0$ as a generic outcome of equations of motion or some energy conditions.

In our analysis and proof of our theorems, especially theorem 2, we assumed that the $\gamma_{d-3}$ part has finite volume. Nonetheless, there are examples of ``EVH black branes'' \cite{EVH-branes} where in the near horizon limit we find AdS$_3$ throats. For these cases one needs to replace horizon area $A_h$ in the definition \eqref{EVH-def-1} by the entropy density of the solution. It would be interesting to extend our theorems for such cases.

Finally, we would like to point out that the near-horizon EVH solutions, and in light of the theorems we proved, may provide  convenient reduction ansatzs to three dimensions. This point was used to construct AdS$_3$ or dS$_3$ solutions within eleven, ten or lower dimensional supergravity theories \cite{Flat-reduction}. It would be nice to see if similar ideas could be useful in constructing dS$_4$ solutions in string theory.


\section*{Acknowledgement}

 MMSHJ and SS  would like to thank Allameh Tabatabaii Prize Grant of Boniad Melli Nokhbegan of Iran. MMSHJ, SS and MHV would like to thank the ICTP network project NET-68.

\bibliographystyle{plain}

\begin{thebibliography}{99}

\bibitem{Hawking-Ellis}
  S.~W.~Hawking and G.~F.~R.~Ellis,
  ``The Large scale structure of space-time,''
  Cambridge University Press (1973).

\bibitem{Stephani}
H. Stephani, D. Kramer, M. A.H. MacCallum, C. Hoenselaers, E. Herlt,
''Exact solutions of Einstein's field equations,''
Cambridge University Press (2003).

\bibitem{Stationary-Black-Holes}
P.~T.~Chrusciel, J.~L.~Costa and M.~Heusler,
``Stationary Black Holes: Uniqueness and Beyond,''
  Living Rev.\ Rel.\  {\bf 15}, 7 (2012)
  [arXiv:1205.6112 [gr-qc]].

\bibitem{ER-review}
  R.~Emparan and H.~S.~Reall,
  ``Black Holes in Higher Dimensions,''
  Living Rev.\ Rel.\  {\bf 11}, (2008) 6,
  [arXiv:0801.3471 [hep-th]].

\bibitem{Kay-Wald-review}
B.~S.~Kay and R.~M.~Wald,
``Theorems on the Uniqueness and Thermal Properties of Stationary, Nonsingular, Quasifree States on Space-Times with a Bifurcate Killing Horizon,''
  Phys.\ Rept.\  {\bf 207}, 49 (1991).

\bibitem{Ferrara-review}
  S.~Ferrara, K.~Hayakawa and A.~Marrani,
  ``Lectures on Attractors and Black Holes,''
  Fortsch.\ Phys.\  {\bf 56}, 993 (2008)
  [arXiv:0805.2498 [hep-th]].

\bibitem{Sen-review}
  A.~Sen,
``Black Hole Entropy Function, Attractors and Precision Counting of Microstates,''
  Gen.\ Rel.\ Grav.\  {\bf 40}, 2249 (2008)
  [arXiv:0708.1270 [hep-th]].

\bibitem{NHEG-phase-space}
  G.~Comp�re, K.~Hajian, A.~Seraj and M.~M.~Sheikh-Jabbari,
  ``Extremal Rotating Black Holes in the Near-Horizon Limit: Phase Space and Symmetry Algebra,''
  arXiv:1503.07861 [hep-th].


\bibitem{KL-papers}
H.~K.~Kunduri, J.~Lucietti and H.~S.~Reall,
  ``Near-horizon symmetries of extremal black holes,''
  Class.\ Quant.\ Grav.\  {\bf 24} (2007) 4169, [arXiv:0705.4214 [hep-th]].

H.~K.~Kunduri and J.~Lucietti, ``A Classification of near-horizon geometries of extremal vacuum black holes,''
J.\ Math.\ Phys.\  {\bf 50} (2009) 082502,  [arXiv:0806.2051 [hep-th]].

P.~Figueras, H.~K.~Kunduri, J.~Lucietti and M.~Rangamani,
  ``Extremal vacuum black holes in higher dimensions,''  Phys.\ Rev.\ D {\bf 78} (2008) 044042  [arXiv:0803.2998 [hep-th]].


\bibitem{KL-review}
  H.~K.~Kunduri and J.~Lucietti,
  ``Classification of near-horizon geometries of extremal black holes,''Living Rev. Rel. {\bf 16} (2013) 8,
  arXiv:1306.2517 [hep-th].

\bibitem{Hollands}
 S.~Hollands and A.~Ishibashi,
  ``All vacuum near horizon geometries in arbitrary dimensions,''
  Annales Henri Poincare {\bf 10} (2010) 1537
  [arXiv:0909.3462 [gr-qc]].


\bibitem{Wald}
R.~M.~Wald, ``Black hole entropy is {N}oether charge'', {\em Phys. Rev.},
  {\bf D48}  (1993) 3427, gr-qc/9307038.

V. Iyer, R.~M.~Wald, ``Some properties of {N}oether charge and a
  proposal for dynamical black hole entropy'', {\em Phys. Rev.}, {\bf D50},
   (1994) 846, gr-qc/9403028.

R.~M.~Wald,
  ``The thermodynamics of black holes,''
Living Rev.\ Rel.\  {\bf 4} (2001) 6,  [gr-qc/9912119].  


\bibitem{NHEG-mechanics}
  K.~Hajian, A.~Seraj and M.~M.~Sheikh-Jabbari,
  ``NHEG Mechanics: Laws of Near Horizon Extremal Geometry (Thermo)Dynamics,''
  JHEP {\bf 1403} (2014) 014
  [arXiv:1310.3727 [hep-th], arXiv:1310.3727].

  K.~Hajian, A.~Seraj and M.~M.~Sheikh-Jabbari,
  ``Near Horizon Extremal Geometry Perturbations: Dynamical Field Perturbations vs. Parametric Variations,''
  JHEP {\bf 1410}, 111 (2014)
  [arXiv:1407.1992 [hep-th]].

\bibitem{EVH/CFT}
M.~M.~Sheikh-Jabbari and H.~Yavartanoo,
``EVH Black Holes, AdS$_3$ Throats and EVH/CFT Proposal,''  JHEP {\bf 1110} (2011) 013  [arXiv:1107.5705 [hep-th]].  

\bibitem{First-law}
  M.~Johnstone, M.~M.~Sheikh-Jabbari, J.~Simon and H.~Yavartanoo,
  ``Extremal black holes and the first law of thermodynamics,''
  Phys.\ Rev.\ D {\bf 88}, no. 10, 101503 (2013)
  [arXiv:1305.3157 [hep-th]].

\bibitem{Bar-Horo}
  J.~M.~Bardeen and G.~T.~Horowitz,
  ``The Extreme Kerr throat geometry: A Vacuum analog of AdS$_2\times$ S$^2$,''
  Phys.\ Rev.\ D {\bf 60}, 104030 (1999)
  [hep-th/9905099].



\bibitem{EVH-2}
R.~Fareghbal, C.~N.~Gowdigere, A.~E.~Mosaffa and M.~M.~Sheikh-Jabbari,
``Nearing Extremal Intersecting Giants and New Decoupled Sectors in N = 4 SYM,''
  JHEP {\bf 0808} (2008) 070
  [arXiv:0801.4457 [hep-th]].

\bibitem{AdS4-EVH}
  R.~Fareghbal, C.~N.~Gowdigere, A.~E.~Mosaffa and M.~M.~Sheikh-Jabbari,
  ``Nearing 11d Extremal Intersecting Giants and New Decoupled Sectors in D = 3,6 SCFT's,''
  Phys.\ Rev.\ D {\bf 81}, 046005 (2010)
  [arXiv:0805.0203 [hep-th]].

\bibitem{Japanese-EVH}
T.~Azeyanagi, N.~Ogawa and S.~Terashima,
  ``Emergent AdS$_3$ in the Zero Entropy Extremal Black Holes,''
  JHEP {\bf 1103} (2011) 004
  [arXiv:1010.4291 [hep-th]].

\bibitem{BTZ-EVH}
  J.~de Boer, M.~M.~Sheikh-Jabbari and J.~Simon,
  ``Near Horizon Limits of Massless BTZ and Their CFT Duals,''
  Class.\ Quant.\ Grav.\  {\bf 28}, 175012 (2011)
  [arXiv:1011.1897 [hep-th]].


\bibitem{EVH-3}
J.~de Boer, M.~Johnstone, M.~M.~Sheikh-Jabbari and J.~Simon,
``Emergent IR Dual 2d CFTs in Charged AdS5 Black Holes,''
  Phys.\ Rev.\ D {\bf 85} (2012) 084039
  [arXiv:1112.4664 [hep-th]].
\bibitem{EVH-4}
 M.~Johnstone, M.~M.~Sheikh-Jabbari, J.~Simon and H.~Yavartanoo,
  ``Near-Extremal Vanishing Horizon AdS5 Black Holes and Their CFT Duals,''
  JHEP {\bf 1304} (2013) 045
  [arXiv:1301.3387].






\bibitem{EVH-Ring}
H.~Golchin, M.~M.~Sheikh-Jabbari and A.~Ghodsi,
``Dual 2d CFT Identification of Extremal Black Rings from Holes,''
  JHEP {\bf 1310}, 194 (2013)
  [arXiv:1308.1478 [hep-th]];
``More on Five Dimensional EVH Black Rings,''  arXiv:1407.7484 [hep-th].  

\bibitem{Hossein-HD}
 H.~Yavartanoo,
  ``EVH black hole solutions with higher derivative corrections,''
  Eur.\ Phys.\ J.\ C {\bf 72}, 1911 (2012)
  [arXiv:1301.4174 [hep-th]].


\bibitem{Hossein-4d}
  H.~Yavartanoo,
  ``On EVH black hole solution in heterotic string theory,''
  Nucl.\ Phys.\ B {\bf 863}, 410 (2012)
  [arXiv:1212.3742 [hep-th]].
  H.~Yavartanoo,
  ``On heterotic black holes and EVH/CFT correspondence,''
  Eur.\ Phys.\ J.\ C {\bf 72}, 2256 (2012).




\bibitem{Hossein-5d}
  H.~Yavartanoo,
  ``Five-dimensional heterotic black holes and its dual IR-CFT,''
  Eur.\ Phys.\ J.\ C {\bf 72}, 2197 (2012)
  [arXiv:1212.4553 [physics.gen-ph], arXiv:1301.3706 [physics.gen-ph]].

\bibitem{SO22}
  S.~Sadeghian, M.~M.~Sheikh-Jabbari and H.~Yavartanoo,
  ``On Classification of Geometries with SO(2,2) Symmetry,''
  JHEP {\bf 1410}, 81 (2014)
  [arXiv:1409.1635 [hep-th]].


\bibitem{DJT}
  S.~Deser, R.~Jackiw and G.~'t Hooft,
  ``Three-Dimensional Einstein Gravity: Dynamics of Flat Space,''
  Annals Phys.\  {\bf 152} (1984) 220.



\bibitem{BTZ}
  M.~Banados, C.~Teitelboim and J.~Zanelli,
``The Black hole in three-dimensional space-time,''
  Phys.\ Rev.\ Lett.\  {\bf 69} (1992) 1849
  [hep-th/9204099].

M.~Banados, M.~Henneaux, C.~Teitelboim and J.~Zanelli,
``Geometry of the (2+1) black hole,''
  Phys.\ Rev.\ D {\bf 48} (1993) 1506
   [Erratum-ibid.\ D {\bf 88} (2013) 6,  069902]
  [gr-qc/9302012].

\bibitem{MP}
  R.~C.~Myers and M.~J.~Perry,
``Black Holes in Higher Dimensional Space-Times,''
  Annals Phys.\  {\bf 172}, 304 (1986).

\bibitem{EVH-branes}
  S.~S.~Gubser and F.~D.~Rocha,
``Peculiar properties of a charged dilatonic black hole in AdS$_5$,''
  Phys.\ Rev.\ D {\bf 81} (2010) 046001
  [arXiv:0911.2898 [hep-th]].

\bibitem{Flat-reduction}
E.~\'{O}.~Colg\'{a}in, M.~M.~Sheikh-Jabbari, J.~F.~V\'{a}zquez-Poritz, H.~Yavartanoo and Z.~Zhang,
  ``Warped Ricci-flat reductions,''
  Phys.\ Rev.\ D {\bf 90}, 045013 (2014)
  [arXiv:1406.6354 [hep-th]].



\end{thebibliography}

\end{document}